%% file: main.tex
\documentclass[12pt, hidelinks]{article}
\usepackage[twoside=true, margin=1in]{geometry}

\newif\iflong 
\longfalse

\usepackage[english]{babel}
\usepackage{bm}
\usepackage{import}
\usepackage[utf8]{inputenc}
\usepackage[english]{babel}
\usepackage{amsmath}
\usepackage{verbatim}
\usepackage{bbm}
\usepackage{dsfont}
\usepackage{amsthm}
\usepackage{amssymb}
\usepackage{enumitem}
\usepackage{mathrsfs}
\usepackage{tikz}
\usepackage{tikzsymbols}
\usepackage{soul}
\usepackage[numbers]{natbib}
\usepackage{hyperref}
\usepackage{xcolor}
\usepackage{color}
\usepackage{caption,graphicx,newfloat}
\graphicspath{ {images/} }
\usepackage{booktabs} 
\usepackage[ruled]{algorithm2e} 
\usepackage{tcolorbox}
\theoremstyle{plain}

\newtheorem{claim}{Claim}[section]

\theoremstyle{remark}

\title{Perpetual Demand Lending Pools}
\author{Tarun Chitra\\ 
        Gauntlet \\
        \texttt{\small tarun@gauntlet.xyz} 
\and Theo Diamandis \\
NYU \\ 
\texttt{\small td2590@nyu.edu} 
\and Nathan Sheng \\
Gauntlet \\
\texttt{\small nathan.sheng@gauntlet.xyz}
\and Luke Sterle \\ 
Gauntlet \\
\texttt{\small luke@gauntlet.xyz} 
\and Kamil Yusubov \\
Gauntlet \\
\texttt{\small kamil@gauntlet.xyz}}
\date{\today}
\input defs.tex

\begin{document}
\maketitle

\begin{abstract}
    Decentralized perpetuals protocols have collectively reached billions of
    dollars of daily trading volume, yet are still not serious competitors on
    the basis of trading volume with centralized venues such as Binance. One of
    the main reasons for this is the high cost of capital for market makers and
    sophisticated traders in decentralized settings. Recently, numerous
    decentralized finance protocols have been used to improve borrowing costs
    for perpetual futures traders.
    These protocols have grown to over \$2.5 billion dollars of assets while generating over \$890 million in fees in 2024. 
    We formalize this class of mechanisms
    utilized by protocols such as Jupiter, Hyperliquid, and GMX, which we
    term~\emph{Perpetual Demand Lending Pools} (PDLPs). We then formalize a
    general target weight mechanism that generalizes what GMX and Jupiter are
    using in practice. We explicitly describe pool arbitrage and expected
    payoffs for arbitrageurs and liquidity providers within these mechanisms.
    Using this framework, we show that under general conditions, PDLPs are easy
    to delta hedge, partially explaining the proliferation of live hedged PDLP
    strategies. Our results suggest directions to improve capital efficiency in
    PDLPs via dynamic parametrization. 
\end{abstract}

\section{Introduction}
Perpetual futures markets are the most liquid trading markets for
cryptocurrencies. These markets facilitate daily volumes on the order of
hundreds of billions of dollars of notional value and are the primary drivers of
hedging and leverage. Although the majority of perpetual futures volume remains
on centralized trading venues, such as Binance or Coinbase, a growing amount of
perpetuals trading occurs on decentralized venues. Such venues, including
Hyperliquid, Jupiter, GMX, and dYdX, have grown to capture almost 10\% of the
daily perpetual future volume~\cite{block-cex-dex-perps}.

Although decentralized venues like dYdX have been live since 2019, the failure
of FTX in 2023 has driven demand for decentralized venues. Exchanges that hold
the state of all balances and positions on a public blockchain can ensure that
user funds are not misappropriated. However, these decentralized exchanges face
a number of challenges that their centralized counterparts do not. The public
nature of positions can lead to front-running and other forms of manipulation.
In addition, the lack of identity verification (which permits traders to walk
away from bad positions) leads to higher collateral requirements for margin
trading.

\paragraph{Market maker loans.} 
To promote liquidity on their platforms, centralized exchanges often offer
market makers loans that can only be utilized on their exchange. In general,
these loans are only partially-collateralized and, as a result, present risks
for other users of the platform. For example, FTX used customer funds to offer a
market maker loan to Alameda Research that did not require any
collateral~\cite{deklotz2023ftx, allen2022testimony, breydo2024contagion}.
Decentralized platforms, by construction, cannot lend out customer funds without
their consent. However, these platforms have higher collateral requirements and
must turn to other mechanisms to provide market maker loans.

\paragraph{Perpetual Demand Lending Pools.}
Decentralized perpetual future exchanges have created novel liquidity pools that
lend to traders on their platform. GMX~\cite{GMX-GLP} first pioneered these
pools, which borrow aspects from both  decentralized spot trading (\eg,
Uniswap~\cite{angeris2021analysis}) and decentralized lending (\eg, Aave and
Compound~\cite{kao2020analysis}). A number of other protocols including
Jupiter~\cite{Jupiter-JLP} and Hyperliquid~\cite{Hyperliquid-HLP} have since
followed suit. These pools allow traders to borrow assets for only one purpose:
to open positions on the associated perpetuals exchange. The protocol liquidates
positions as soon as they become undercollateralized. Since these loans resemble
demand loans~\cite{jones1980valuation}, we call these mechanisms \emph{Perpetual
Demand Lending Pools} (PDLPs).

\paragraph{PDLP mechanics.}
Although perpetual demand lending pools have different implementation details
with regard to how their lending functions operate, they share common traits:
\begin{itemize}
    \item Users pool together funds into a liquidity pool. The protocol only
    permits certain assets in this pool and maintains a target composition.
    \item Traders, including market makers, borrow from the pool to open 
    positions on the associated exchange.
    \item LPs earn fees from traders who borrow from the pool. These fees are
    proportional to the size of the position and are paid at regular intervals.
    \item Arbitrageurs ensure the pool stays at the protocol's desired target
    composition.
    \item LPs, rather than the protocol itself, realize losses when the pool
    cannot rebalance or when the protocol cannot liquidate underwater positions
    quickly enough.
\end{itemize}

\begin{figure}
    \centering
    \includegraphics[width=\linewidth]{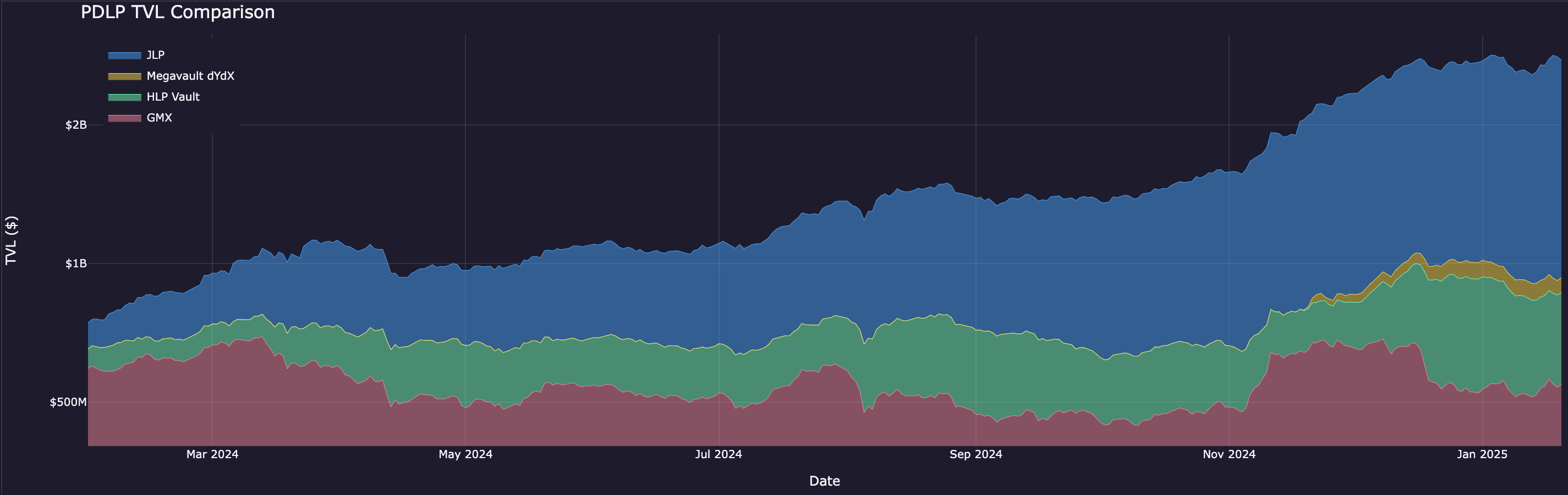}
    \caption{Total Value Locked for PDLPs  (\ie~assets locked into PDLP pools)
    in 2024 (\href{https://app.hex.tech/a22e666c-4bc7-40d3-93ed-45a25554a14e/hex/7edb6b2b-a2d5-499b-a457-62a5106cffe0}{Link To Data}).}
    \label{fig:pdlp-tvl}
\end{figure}

\paragraph{Background on existing PDLPs.}
PDLPs have amassed over \$2.5 billion in assets and generated roughly \$897.73
million dollars in aggregate fees\footnote{ As of December 27, 2024, the
cumulative fees generated by each protocol are: Jupiter (\$445.48
million~\cite{jupiter-perps-fees}), dYdX MegaVault (\$4.25
million~\cite{gauntlet-dydx}), HLP (\$58 million
\cite{Hyperliquid-HLP-vault}), GMX v1 + v2 (\$390
million~\cite{gmx-perps-fees}) } since GMX launched in September 2021. Jupiter's
JLP pool and Hyperliquid's HLP pool followed, each making changes to the
original GMX PDLP mechanism. In the last year, the total value of assets in
PDLPs has grown from under \$750 million to over \$2.5 billion, and these
positions have yielded a return of approximately 40\% in 2024. We note that the
rising popularity of PDLPs has tracked the growth of perpetuals trading on
decentralized exchanges.

\paragraph{Delta hedging.}
PDLP LPs have had success delta hedging in practice. These LPs, who provide
assets to the pool, can take derivative positions to offset their exposure to
risky asset price movement (the delta) while still capturing fees. This property
allows LPs to enter and exit positions without realizing large convexity losses
due to volatility~\cite{hull2017optimal} and, as a result, facilitates cheaper
market making~\cite{huh2015options, stoikov2009option}. We note that many
previous attempts in decentralized finance to provide delta-hedged portfolios
have been unsuccessful in practice (\eg, see~\cite{angeris2023replicating,
khakhar2022delta, lambert2022panoptic}). This observation begs the question:
Why are PDLPs easier to hedge? We characterize PDLP mechanics and shed light on 
this question in this paper.

\subsection{This paper.}
We provide the first formalized (to the authors' knowledge) definitions and
description of PDLPs. We aim to formalize the financial properties of these 
pools and explain their large amount of fee revenue.

In~\S\ref{sec:pdlp}, we define PDLPs and describe how perpetuals exchanges use
these pools. We then describe single-period arbitrage of the perpetual funding
rate---the core mechanism that ensures the futures contract and the spot asset
have consistent prices. We also describe arbitrage opportunities in the PDLP
itself, which are similar to those in constant function market makers. We give
necessary conditions on the PDLP fee so that both the funding rate arbitrageur
and the PDLP liquidity providers are profitable. The arbitrage loss faced by LPs
resembles loss-versus-rebalancing (LVR) from the CFMM
literature~\cite{milionis2022automated}. 

In~\S\ref{sec:weight-based-arb}, we formalize mechanisms used by PDLPs to stay
near a target portfolio, which are usually specified in terms of desired
relative weights, These mechanisms, called target weight mechanisms, provide
economic incentives for liquidity providers to add or remove assets so that the
PDLP maintains a target portfolio. Unlike CFMMs, PDLP liquidity providers may
deposit (or withdraw) any subset of assets to (or from) the pool. This feature
creates novel opportunities for arbitrageurs which resemble create-redeem
arbitrage in exchange traded funds~\cite{pan2017etf}. We model these arbitrage
opportunities as optimization problems, which can be approximately solved in
practice. Using this model, we demonstrate that the target weight mechanism
bounds an LP's delta exposure within in PDLPs. This fact may explain the
proliferation of successful PDLP hedging strategies, especially for those
utilizing Jupiter's JLP~\cite{Gauntlet-hJLP}. We note that hedging LVR in CFMMs
is possible, but these protocols have far less usage in
practice~\cite{defi-lp-do-not-worry, lipton2024unified, lambert2022panoptic}.

Finally, in~\S\ref{sec:hedged-pdlp}, we describe dynamic delta hedging of risk
assets in a PDLP using mean-variance optimization. We give simple sufficient
conditions for when the delta hedge improves the Sharpe ratio of PDLP LP returns.
We also describe when a single PDLP pool should be split into two pools to 
improve the delta-hedged Sharpe ratio.

\paragraph{Notation.}
We denote the unit simplex as $S^n = \{ (w_1, \ldots, w_n) \in \reals^n_+ :
\sum_i w_i = 1\}$. We use $\reals^n_+$ for the nonnegative orthant and
$\reals^n_{++}$ for the positive orthant in $n$ dimensions. We denote set of
natural numbers from $1$ to $k$ is denoted as $[k] = \{1, 2, \ldots, k\}$. The
$L^p$ norm of a vector $x \in \reals^n$ is denoted as $\|x\|_p = \left(\sum_i
|x_i|^p\right)^{1/p}$.

\section{Perpetual Demand Lending Pools}\label{sec:pdlp} 
In this section, we formalize a model of Perpetual Demand Lending Pools (PDLPs).
Note that we ignore implementation details of practical mechanisms that do not
impact economic payoffs. We begin by describing a simplified linear model of a
perpetuals exchange. We refer the interested reader
to~\cite{ackerer2024perpetual} for models that account for additional parameters
found in practice.

\subsection{Perpetuals exchange}
A perpetuals exchange facilitates perpetual future contracts between two groups
of users: those holding long positions and those holding short positions in some
underlying asset. In these contracts, the group of users with the larger
cumulative position pays a fee, called the \emph{funding rate}, at regular
intervals to the other group. This fee incentivizes the two positions to be
roughly equal. Arbitrageurs who observe an imbalance between the long and short
positions can equalize the positions and earn a profit if the price moves in the
same direction as the imbalance.

Concretely, we define a perpetual future contract for some asset as a triple
$(L, S, p_0) \in \reals_+^3$, where $L$ and $S$ are respectively the cumulative
long and short positions, and $p_0$, is the \emph{mark price} of the underlying
asset, which is set at the beginning of each time interval. The funding rate is
a function of these three parameters, and of the price of the underlying asset
at the end of this interval. We can interpret the mark price $p_0$ as the price
at which two parties enter into the perpetual contract to bet on a price change
during the time interval. (For example, large centralized exchanges such as
Binance, OkEx, and Coinbase update this mark price every eight
hours~\cite{ackerer2024perpetual}.) Arbitrageurs will ensure that, at the
beginning of the next interval, the mark price is equal to the price of the
underlying asset. A perpetuals exchange contains perpetual future contracts for
some collection of $n \in \mathbb{N}$ assets.

\paragraph{Trades.}
A user trades on a perpetuals exchange by placing collateral on the exchange and
then  opening a long or short position in some asset. We will assume, for
simplicity, that the collateral is a num\'eraire (such as US dollars) whose
price at all times is 1. A \emph{trade} is defined as a tuple $(c, \delta, \eta,
p_0) \in \reals_+ \times \reals_+ \times (\reals -\{0\}) \times \reals_+$ where
$c$ is the collateral the user places, $\delta$ is the position size, $\eta$ is
the leverage level, and $p_0$ is the price when the trade is created. A position
is a short position if $\eta < 0$ and has leverage if $|\eta| > 1$. At the time
of creation, the position requires collateral $c$ that satisfies the
\emph{collateral condition}:
\begin{equation}\label{eq:collateral-condition}
    p_0\delta \le \lvert{\eta}\rvert c.
\end{equation}
This means that a user can create a notional position with size $p \delta$ that
is up to $|\eta|$ times larger than the collateral they deposit.\footnote{This
condition is what is often meant when the claim that perpetual futures are more
capital efficient for leverage than traditional lending protocols is made,
see~\eg,~\cite{ackerer2024perpetual}.}

\paragraph{Liquidations.}
A trader must ensure that the trade has sufficient collateral as the mark price
of the risky asset $p$ fluctuates, otherwise the exchange will liquidate this
position. Define the \emph{liquidation condition} as 
\[
    \mathbf{sign}(\eta) \delta (p_0 - p) \ge c.
\]
In words, if the the loss in the perpetuals position due to price change is
greater than the collateral value, the position will be liquidated. For example,
a long position is liquidated once the price of the risky asset drops below $p_0
- c/\delta$. If this position was created with minimal collateral,
saturating~\eqref{eq:collateral-condition}, then the position is liquidated when
the price drops below $p_0(1 - 1/\eta)$.

\paragraph{Example.}
Consider a trade \( T = (c, \delta, \eta, p_0) \) with collateral $c = \$2000$,
position size $\delta = 1$ ETH, leverage $\eta = 4$ and initial ETH price $p_0 =
\$2000$. The initial notional value of the position, $\delta p_0$, is 8,000 USD.
This trade satisfies the collateral condition~\eqref{eq:collateral-condition}.
As the price of ETH $p$ decreases, the position remains valid until the liquidation
condition is met, which occurs when $p = \$1,5000$. Once this price is reached,
the perpetuals exchange will liquidate this position.

\paragraph{Funding rate.}
We consider the linear funding rate 
\begin{equation}\label{eq:funding-rate}
    \gamma_L(L, S, p, p_0) = \kappa \left(\frac{L}{S} - \frac{p}{p_0}\right),
\end{equation}
where $p$ is the price of the underlying asset and $\kappa > 0$ is a constant.
Note that on-chain perpetuals exchanges typically use a price oracle to
determine the price $p$. This price oracle is often (but not always) the median
of the prices on a set of centralized exchanges with sufficient liquidity. In
practice, most exchanges use the linear funding rate model, although a number of
other models exist in the literature (for examples,
see~\cite{ackerer2024perpetual, angeris2023primer, he2022fundamentals}). If the
funding rate is positive ($\gamma_L > 0$), then the users with short positions
pay those with long positions, and if the funding rate is negative $(\gamma_L <
0$), then the users with long positions pay those with short positions.

\paragraph{Example.}
Consider an ETH-USD perpetual futures contract $(L, S, p_0)$ with cumulative long
position $L = 1000$ ETH, short position $S = 250$ ETH, and mark price $p_0 =
2000$ USD, the funding rate is $\gamma_L = \kappa\left(\frac{L}{S} - \frac{p}{p_0}\right)$.
When $p = p_0$, including at the start of the contract period, this equation simplifies to:  
\[
    \gamma_L = \kappa\left(\frac{1000}{250} - 1\right) = 3\kappa
\]   
Here, longs pay shorts, which incentivizes traders to take short positions until
$\gamma_L = 0$.

\subsection{Perpetual Demand Lending Pools}\label{sec:pdlp-desc}
A perpetual demand lending pool allows traders to borrow
assets to open positions on an associated perpetuals exchange. Yield-seeking
depositors, called \emph{liquidity providers} (LPs), deposit these assets into
the pool. Traders can then use these assets to collateralize their positions on
the perpetuals exchange, for which they pay fees to the liquidity providers.
Since most of these positions are levered, liquidity providers effectively
provide under-collateralized loans which only may be used for trading on the
associated perpetuals exchange.

\paragraph{Definition.}
Formally, we specify a Perpetual Demand Lending Pool (PDLP) by its current
portfolio $R \in \reals^n_+$, target portfolio $\pi \in \reals^b_+$, lending fee
$f \in (0,1)$, and outstanding loans $\{c_i\}_{i \in [M]}$, where each loan $c_i$ is the
collateral used by position $i = 1, \dots, M$. Note that one can view the loans
made by the PDLP as collateralizing an equal but opposite position to the one
desired by the trader, much as a market maker takes the opposite size of a
leveraged trade. The trader associated with position $i$ must pay a fee $fc_i$
for each time period that they hold a loan. We note that the target portfolio
and fees may be dynamically updated based on external market conditions. The
total value lent to traders must be less than the total value of the pool's
assets, \ie,
\[
    \sum_{i=1}^M c_i \leq R.
\]
If any position $i$ becomes invalid, the protocol 
liquidates this position and returns the collateral $c_i$ back to the PDLP.

\paragraph{Dynamics.}
Using the definitions above, we model how trades interact with the perpetuals
exchange and how LPs utilize the PDLP via a fixed sequence of transaction
execution. First, note that we assume that there exists an unmanipulable price
oracle that submits a price update for the $n$ assets at each block. The
sequence of interactions executed by a solvent, collateralized PDLP is:
\begin{enumerate}
    \item The price oracle is updated.
    \item Liquidatable positions are removed.
    \item Funding rates, based on the positions opened at the previous period
    and the oracle price, are paid out.
    \item Fees for loans used in the previous period are paid out to LPs from
    trader positions.
    \item LPs update their PDLP portfolios.
    \item Traders submit new trades and valid trades are executed.
\end{enumerate}
In Appendix~\ref{app:dynamics}, we formalize these dynamics in terms of the
state variables defined so far. We also note that the dynamics here do not
include any constraints on updates to the target portfolio since different
protocols use different mechanisms to update the target portfolio. We study the
mechanism used by GMX and Jupiter to adjust and/or realize a target portfolio
in~\S\ref{sec:weight-based-arb}.

\paragraph{Liquidity creation and redemption.}
When a liquidity provider (LP) deposits assets into a PDLP, they receive tokens
which entitle them to a pro-rata claim on the pools' portfolio and fees. An LP
may deposit or withdraw any subset of assets into the PDLP. Note that this
mechanism differs from that of a constant function market maker (CFMM), where
LPs generally must deposit or withdraw all assets and/or define a price range
over which their liquidity can be used. (See~\cite{angeris2022constant,
bar2023uniswap, fan2021strategic} for details.) PDLPs also allow swaps between
the assets in the pool, similarly to a CFMM. Unlike a CFMM, a PDLP typically
sets swap prices using an external price oracle and charges a dynamic fee chosen
to maintain the PDLP's target portfolio. LPs can create a share of a PDLP,
entitling them to a pro-rata claim on the pool's assets and fees, by tendering
any of the assets held by the PDLP. Similarly, they can redeem a PDLP share for
any valid portfolio that has the same value as the share.

\subsection{Arbitrage in a single period}\label{subsec:arb-single} 
We consider two arbitrage opportunities created by a price movement in the
underlying asset: funding rate arbitrage and PDLP share arbitrage.
We will show that both the pool liquidity providers, who deposit assets prior to
the funding rate change, and funding rate arbitrageurs are profitable as long as
the PDLP fees remain inside a particular interval.

Assume that, at the beginning of the period, the price of the underlying asset
is $p_0$ and the funding rate is zero, so long and short positions have the same
size: $L = S = L_0$ (see~\eqref{eq:funding-rate}). When the price of the
underlying asset moves from $p_0$ to $p$, this price movement opens two
arbitrage opportunities:
\begin{enumerate}
    \item A funding rate arbitrageur borrows from the PDLP pool and opens a long
    or short position to capture the funding rate until the mark price and true
    price converge.
    \item A PDLP arbitrageur creates or redeems shares of the PDLP to capture
    the spread between the value of the PDLP's assets and the LP token's spot price.
\end{enumerate}
We describe each of these opportunities and bound the loss of the PDLP. We
assume that $p > p_0$, but an entirely symmetric derivation gives the case when
$p < p_0$.

\subsubsection{Funding rate arbitrage.}
The funding rate~\eqref{eq:funding-rate} after price movement is
\[
\gamma_L(L_0, S_0, p, p_0) = \kappa \left(\frac{L_0}{S_0} - \frac{p}{p_0}\right)
    = \kappa\left(1 - \frac{p}{p_0}\right) < 0.
\]
Thus, the users with short positions must pay those with long positions. A
funding rate arbitrageur opens a long position of size $\ell$ to capture this
funding rate prior to the next funding payment made by the PDLP. This arbitrage
exists as long as the funding rate continues to be non-positive; thus, the
largest long position that can capture this funding rate (\ie, the $\ell$ that
makes $\gamma_L(L_0 + \ell, S_0, p, p_0) = 0$) is given by
\begin{equation}\label{eq:long-val}
    \ell = L_0\left(\frac{p}{p_0} - 1\right).
\end{equation}
After opening this position, the arbitrageur receives a pro-rata share of the
funding rate: $\ell/(L_0 + \ell) \cdot \left\vert \gamma_{L}(L_0, S_0, p,
p_0)\right \vert$. To open this position, the arbitrageur must pay a fee
$f\cdot\ell$. Thus, the arbitrageur makes a profit when the fee is less than the
revenue,~\ie
\[
f \ell \leq \frac{\kappa \ell}{L_0 + \ell} \left( \frac{p}{p_0} - 1\right)
= \frac{\kappa \ell}{L_0}\left(1 - \frac{p_0}{p}\right)
\]
where we utilized~\eqref{eq:long-val} in the equality. 
If the relative price increment is bounded, $1 < \frac{p}{p_0} \leq B$, then we
have that
\begin{equation}\label{eq:fee-upper-bound}
f \le \frac{\kappa}{L_0}\left(1 - \frac{p_0}{p}\right) 
\le \frac{\kappa(1 - B^{-1})}{L_0}.
\end{equation}
This means that a fee that is inversely proportional to the size of the 
cumulative long position (the open interest) ensures that the funding rate
arbitrage is profitable. As the open interest increases, this fee must decrease
(or the funding rate must increase) to offset the increase in arbitrage size.

\paragraph{Discussion.}
We note that the assumption $1 \leq \frac{p}{p_0} \leq B$ is not restrictive
when compared to other forms of optimal fees (\eg,~the assumptions of LVR with
fees~\cite{milionis2023effect}). If one has a model for the price process (such
as a geometric brownian motion, which is used in LVR), one can use Chebyschev's
inequality to get a bound of the form $\Prob[\frac{p}{p_0} \leq B] \leq
f(\sigma)$ for an increasing function $f$ of the price process volatility
$\sigma$. In short, one can turn a price process into a pointwise bound
$\tfrac{p}{p_0} \leq B$ that holds with high probability.

\subsubsection{PDLP arbitrage.}
The same price movement also introduces a discrepancy between the price at which
the PDLP is willing to exchange assets and the external market price. An
arbitrageur may buy the asset at the old price $p_0$ (which need not be the mark
price) on the PDLP and sell it on an external market for the new, higher
price $p$. The price impact of this trade on both markets and the associated
fees bound the size of this trade. For simplicity, we focus only on the price
impact of buying on the PDLP. (We will return to the PDLP's fees
in~\S\ref{sec:weight-based-arb}.) We model the PDLP's swap functionality using
the forward exchange function $G: \reals_+ \to \reals_+$, which is nonnegative,
concave, and nondecreasing~\cite{angeris2022constant}. This function denotes,
for a quantity $x$ of the num\'eraire, the amount of the asset, $G(x)$ that the
PDLP is willing to sell. We assume that $G(0) = 0$ (no free lunch) and that the
PDLP initially quotes at the old price: $G'(0) = 1/p_0$. If the arbitraguer uses
$x$ units of the num\'eraire to purchase the asset on the PDLP, their profit is
\[
    pG(x) - x.
\]
When $G$ is differentiable, the most profitable trade easily follows from the
first order conditions:
\[
    G'(x^\star) = 1 / p
\]
The net change in the PDLP's value is then
\[
    \underbrace{x^\star}_{\text{num\'eraire in}} 
    - \;\;p \cdot \underbrace{G(x^\star)}_{\text{asset out}}.
\]
This arbitrage is equivalent to the standard price arbitrage in a
CFMM~\cite{angeris2020improved}, and the LP loss is analogous to
loss-versus-rebalancing in CFMMs~\cite{milionis2022automated}.

\subsubsection{Choosing fees}
After arbitrage, PDLP LPs have earned a fee from the
PDLP borrow but also suffered a rebalancing loss from the pool.
Their profit is 
\[
    \mathsf{Profit} = \underbrace{f \ell}_{\text{fee revenue}} 
    + \;\;
    \underbrace{x^\star - p G(x^\star)}_{\text{rebalancing cost}}.
\]
Note that the second term is negative since
\[
    x - p G(x) \le x - p G'(0)(x - 0) = x(1 - p/p_0) < 0,
\]
for any $x$, where the first inequality follows from the definition of concavity
and the second from the fact that $p > p_0$. Using~\eqref{eq:long-val} and the
definitions above, this profit is nonnegative whenever
\begin{equation}\label{eq:fee-lower-bound}
    f \ge \frac{1}{L_0} \cdot \frac{G(x^\star) - x^\star / p}{G'(0) - 1/p}
\end{equation}
Using the definition of concavity, we have the following sufficient condition 
for the fee lower bound~\eqref{eq:fee-lower-bound} to hold:
\[
    f \ge \frac{x^\star}{L_0}.
\]
In other words, LPs are profitable as long as there is sufficient liquidity and
the PDLP swap functionality, described by the price impact function, ensures the
arbitrage size is sufficiently small. A stronger condition follows from
additional assumptions on the forward exchange function $G$ (for example, strong
concavity).

Combining the upper bound~\eqref{eq:fee-upper-bound} and lower
bound~\eqref{eq:fee-lower-bound}, we have conditions for fees that ensure both
LPs make a profit and arbitrageurs arbitrage the funding rate, which balances
the long and short positions. Specifically, if price changes are bounded then
the PDLP can set parameters so that fees inversely proportional to the open
interest (\ie, $f = \Theta(1 / L_0)$) ensure that both LPs and funding rate
arbitrageurs are profitable. This result suggests that PDLPs can have a
sustainable equilibrium between traders and LPs under mild conditions. Moreover,
these results suggest that PDLPs with dynamic fees (\ie, where the fee $f$
depends on the cumulative long and short positions) are more likely to realize
this equilibrium.

\paragraph{Example.}
Consider a PDLP with a forward exchange function that resembles Uniswap v2:
\[
    G(x) = \frac{R_2 x}{R_1 + x}.
\]
The forward exchange rate is $G'(x) = R_1R_2/(R_1 + x)^2$, and the pool 
initially quotes the original price: $R_2/R_1 = 1/p_0$. The LP loss can be
calculated as
\[
    R_1 + pR_2 - 2\sqrt{pR_1R_2}
\]
Using the bounds on the price change, we can show that the LP is profitable as
long as
\[
    f \ge \frac{(B - 1)R_1}{L_0}.
\]
The fee must increase as the size of the PDLP increases, but it may decrease as
the amount of open interest increases.


\section{Weight-based arbitrage}\label{sec:weight-based-arb} 
Perpetual demand lending pools typically have a target portfolio that they aim
to maintain. For example, GMX's GLP and Jupiter's JLP pool aim to maintain a
constant relative composition of assets. These target portfolios ensure the
lending pools stay diversified and maintain inventory across the different
assets used for collateral. To maintain these target portfolios, PDLPs provide
economic incentives to liquidity providers and traders to rebalance the pool. We
describe the mechanism for distributing these incentives, which resemble PID
controls, and the associated arbitrage problems in this section.

\subsection{Target Weight Mechanisms}\label{sec:twm}
A \emph{target weight mechanism} (TWM) for a PDLP takes a
target portfolio and attempts to minimize the deviation between the PDLP's 
current portfolio and this target portfolio.
The PDLP constructs this target portfolio to limit its exposure to a single
asset or a small set of correlated positions,\footnote{ We suspect that the PDLP
can choose this target portfolio with approximation algorithms akin to those
proposed for decentralized
lending~\cite{bastankhah2024thinking,nadkarni2024adaptive,bertucci2024agents}.}
and this target portfolio also ensures that liquidity providers add diversified
assets to the pool to meet demand. Typically, PDLPs specify this portfolio in
terms of a target price-weighted relative asset composition. We consider this
case here, but our analysis easily extends to considering the target portfolio
directly.

\paragraph{Utilization.}
Unlike a CFMM, the PDLP has some set of assets that cannot be redeemed: those
pledged as collateral for a perpertual future position. We say that these assets
are utilized. We define the available portion of the pool, $R^A$, as 
\[
    R^A = R - \sum_{i=1}^M c_i.
\]
By construction, this vector is nonnegative; the pool itself does not use 
leverage.

\paragraph{Target weight.}
Given a PDLP with reserves $R \in \reals^n_+$ and asset prices $p \in \reals^n_{++}$,
the weight of a PDLP is the price-weighted relative composition of the assets in the pool:
\[
w(p, R) = \frac{p \odot R}{p^T R},
\]
where $\odot$ is the element-wise (Hadamard) product. We normalize by the
portfolio value, $p^TR$. The PDLP's TWM attempts to minimize the deviation
between this weight and a target weight $w^{\star}$ by providing economic
incentives to arbitrageurs to add or remove assets\footnote{ The target weight
is typically set to reflect expected loan quantities. See, for
example,~\cite{jlp-weight-proposal}. }. By providing these incentives, the TWM
essentially aims to (indirectly) solve the optimization problem
\begin{equation}\label{eq:protocol-opt-problem}
    \begin{aligned}
        &\text{minimize} && \left\Vert w(p, R + \Delta) - w^{\star}\right\Vert,\\
        &\text{s.t.} && \Delta \geq - R^A,
    \end{aligned}
\end{equation}
with variable $\Delta \in \reals^n$ denoting the assets to be added or removed
from the PDLP. The constraint indicates that the amount removed cannot be in
excess of the unutilized assets $R^A$. The act of providing incentives to
approximately solve~\eqref{eq:protocol-opt-problem} resembles both the Hedge
Algorithm~\cite{freund1997decision} and other online learning algorithms on the
simplex~\cite{hazan2016computational}. We suspect that weight update rules could
be designed to be no-regret, as they are in resource markets
in~\cite{angeris2024multidimensional,diamandis2023designing}. We leave
formalizing this connection to future work.

\paragraph{Incentives and the discount rate.}
PDLPs give liquitity providers of underweight assets a discounted pro-rata
ownership of the PDLP. (Equivalently, these depositors receive a subsidy,
denominated in the pool's portfolio, for depositing these underweight assets.)
Let $\Delta \in \reals^n$ denote an LP's proposed change to the PDLP reserves.
If this update is valid, \ie, $-\Delta \ge R^A$, then the PDLP's reserves are
updated to $R + \Delta$ and this LP recieves the pro-rata ownership in the pool
\[
    (1 + F(p, w^{\star}, R, \Delta)) \cdot \frac{p^T\Delta}{p^T(R + \Delta)},
\]
with \emph{discount rate} $F : \reals^n_+ \times S^n \times \reals^n_+
\times \reals^n \rightarrow \reals$. LPs who move the weights closer to the
target receive a discount on the pool share price, \ie, $F > 0$. In particular,
we assume the following conditions on $F$ hold:
\begin{itemize}
    \item There is no discount for a zero trade: $F(p, w^{\star}, R, 0) = 0$ for
    all $p, w^{\star}, R$.
    \item The discount rate $F$ is concave in its last argument $\Delta$.
    \item The discount rate is maximized by trades that achieve the target
    weights: when $\Delta \in \argmax_{\delta} F(p, w^{\star}, R, \delta)$, we
    have $w(p, R+\Delta) = w^{\star}$ and $w(p, R) \neq w^{\star}$.
\end{itemize}
The first condition states that there is a constant discount at all prices,
weights, and reserves for an empty trade. The second condition states that the
discount rate has constant or diminishing growth. And the final condition states
that the discount is maxmized for trades that achieve the target weight.
Existing PDLP LPs implicitly pay the new pool LP via dilution when $F > 0$, and
this new LP dilutes the existing LPs by a factor of $1/(1 + F)$. (Alternatively,
existing LPs receive a subsidy when $F < 0$.)

\paragraph{Target weight arbitrage.}
TWMs rely on discount rate arbitrage to ensure that PDLPs remain near their
target weights. The discount rate $F$ must not only incentivize arbitrauers to
solve~\eqref{eq:protocol-opt-problem} but also not excessively dilute existing
LPs. This arbitrage mechanism resembles that of exchange-traded funds with
target portfolios within traditional finance (see
\cite{petajisto2017inefficiencies} and citations within). However, the TWM
introduces a variable create-redeem fee that can introduce arbitrage
opportunities even when the price of the PDLP share does not differ from the
value of the pool's assets. We construct the relevant optimization problem in
the case of share creation and of share redemption.

\paragraph{Example.}
Here, we walk through a simple example to demonstrate the benefits of the target
weight mechanism. Consider a PDLP with two assets, the num\'eraire and a risky
asset, that are equally weighetd. Assume that all of the risky asset is being
used as collateral for perpetual positions. What happens when the price of the
risky asset changes? If the risky asset decreases in price, the PDLP will be
underweight the risky asset. Arbitrageurs will be incentivitized to add more of
this asset to the pool, which will allow for more positions that use the risky
asset as collateral to be opened. On the other hand, if the risky asset
increases in price, the PDLP will be overweight the risky asset. Arbitrageurs
can either close positions in the money, redeeming the risky asset, or add
num\'eraire to the pool.

\subsection{TWMs in practice}
Here, we review several of the most popular target weight mechanisms used in
practice. We note that PDLPs often have differing requirements for which
colleteral traders may use for which positions.

\paragraph{GMX GLP pool and Jupiter JLP pool.} GMX introduced the first PDLP,
launched on the Arbitrum blockchain, in September 2021~\cite{GMX-GLP}. This PDLP
provided incentives for users to maintain a target weight: changes to the pool
that brought the asset composition closer to the target received a rebate,
whereas changes that pushed the asset composition further from the target
incurred a fee. As implemented, the mechanism resembles a simple
proportional-integral-derivative (PID) controller mechanism that requires
arbitrageurs to implement the control policy~\cite{gmx-repo}. Jupiter's JLP pool
introduced a similar mechanism but with changes to reduce the funding rate risk
borne by LPs using delta-hedge strategies. As of December 2024, GMX's PDLP has
approximately \$650 million in assets~\cite{defillama-gmx} and supports \$10
million in open interest (dollar notional value of open positions). Jupiter's
PDLP is significantly larger, with roughly \$1.5 billion of assets and \$870
million of open interest~\cite{gauntlet-jupiter}. 

\paragraph{Hyperliquid HLP pool.}
Hyperliquid, a large exchange with over \$3 billion of open interest in December
2024, similarly uses a PDLP but with a much more opaque mechanism. The
Hyperliquid PDLP uses a closed-source strategy, managed by a single whitelisted
entity, to adjust its target portfolio and lend to
traders~\cite{Hyperliquid-HLP}. This PDLP still allows for permissionless
liquidity provision and distributes its rewards amongst the LPs. As of December
2024, the Hyperliquid PDLP holds approximately~\$350 million of
assets~\cite{defillama-hlp}.

\paragraph{dYdX MegaVault.}
The dYdX MegaVault~\cite{dydx-megavault} PDLP uses a fixed, open-source market
making strategy (see~\cite{dydx-repo} for details). This strategy depends on
certain parameters, which are periodically adjusted by a single whitelisted
actor. These public strategies are similar to classical market-making
strategies~\cite{avellaneda2008high} and generally perform well in
lower-frequency and liquidity markets. As of December 2024, the dYdX MegaVault
PDLP holds approximately~\$65 million of assets~\cite{gauntlet-dydx}.

\subsection{Share creation}\label{subsec:creation} When there is a positive
discount, $F > 0$, an arbitraguer can create a PDLP share worth more than the
assets she provides to the PDLP. Recall from~\S\ref{sec:pdlp-desc} that, an LP
who deposits a portfolio $\Delta \in \reals^n_+$ receives a pro-rata ownership
in the PDLP worth
\[
    (1 + F(p, w^{\star}, R, \Delta)) \cdot \frac{p^T\Delta}{p^T(R + \Delta)} \cdot V(p, R + \Delta).
    = (1 + F(p, w^{\star}, R, \Delta)) \cdot p^T \Delta.
\]
This costs the LP $p^T \Delta$ in assets. Thus, the maximum value arbitrage is
given by the solution to the optimization problem
\begin{equation}\label{eq:creation-arb}
\begin{aligned}
    &\text{maximize}&& F(p, w^{\star}, R, \Delta) (p^T \Delta),\;
    \text{s.t.} && \Delta \geq - R^A.
\end{aligned}
\end{equation}
It is clear that share creation is only rational for the arbitrageur when there
is a positive discount. We claim that if $F$ is an admissible discount rate and
is concave in $\Delta$ and decreasing, then this optimization problem's maxima
can be approximated when $F$ is strongly concave. 

\paragraph{Approximate optimum.}
Often, PDLPs do not disclose the exact form of the discount function $F$.
Instead, arbitrageurs have black-box access to the function $F$ and can query
the discount for a given trade size $\delta$. We demonstrate that, under certain
assumptions on $F$, arbitrageurs can approximately solve the creation arbitrage
problem~\eqref{eq:creation-arb} with only this black-box discount access. This
fact suggests that arbitrageurs can still successfully trade on closed-source
PDLPs, such as Jupiter's JLP. 

\paragraph{Surrogate function.}
Suppose that $F$ is $\mu$-strongly concave~\cite[\S9.1]{boyd2004convex} in
$\delta$. That is, if we fix $p, w^{\star}, R$ and let $f(\delta) = F(p,
w^{\star}, R, \delta)$, then $f(\delta)$ satisfies
\begin{equation}\label{eq:strong-oncavity}
    f(\delta) 
    \leq 
    f(\delta') + \nabla f(\delta')^T(\delta-\delta') - \frac{\mu}{2} \Vert \delta-\delta'\Vert_2^2
\end{equation}
for all $\delta'$ in the domain of $f$. Moreover, since $F(p, w^{\star},
R, 0) = 0$ for all $p, w^{\star}, R$, we have that
\[
    f(\delta) \leq g^T\delta - \frac{\mu}{2} \delta^T\delta = h(\delta)
\]
for any subgradient $g \in \partial F(p, w^{\star}, R, 0)$. This also implies
that $f(\delta) \leq g^{T} \delta$ for any $g \in \partial f(0)$. We define
$\delta_H^{\star} = \min_{\delta} h(\delta)$, which can be easily computed as
\[
\delta_H^{\star} = \frac{1}{\mu} g
\]
Using the strong convexity of $f(\delta)$, we can show that $\delta_H^{\star}$
is close to $\delta_F^{\star} \in \argmax_{\delta} f(\delta)$, as illustrated in
the following claim:
\begin{claim}
    Suppose $f(\delta) = 0$, $f(\delta)$ is $\mu$-strongly concave, and
    $\max_{\delta} \Vert \nabla f(\delta) \Vert_2^2 \leq G$. Then 
    \[
    \|\delta_F^{\star} - \delta_H^{\star}\|_2 \leq 2G/\mu.
    \]
    If the constraint is not active, then we have a tighter bound: $\|\delta_F^{\star} - \delta_H^{\star}\|_2 \leq \sqrt{G}/\mu.$
\end{claim}
\begin{proof}
    Recall the first-order condition for strong concavity~\cite[\S9.1]{boyd2004convex}:
    \[
        f(\delta_F^\star) \le f(\delta_H^\star) + \nabla f(\delta_H^\star)^T(\delta_F^\star - \delta_H^\star) - \frac{\mu}{2}\Vert \delta_F^\star - \delta_F^\star\Vert_2^2
    \]
    The first bound follows from using the fact that $0 \le f(\delta_F^\star) -
    f(\delta_H^\star)$ and Cauchy-Schwarz. For the second bound, recall that a
    $\mu$-strongly convex function $g : \reals \rightarrow \reals$ also
    satisfies the Polyak-Lojasiewicz condition:
    \[
    \frac{1}{2}\Vert \nabla g(x)\Vert^2_2 \geq \mu(g(x)-g(x^{\star}))
    \]
    where $x^{\star} \in \argmax_x g(x)$.
    Since $-f$ is $\mu$-strongly convex, this implies that
    \[
    f(\delta_F^{\star}) - f(\delta_H^{\star}) 
    \leq \frac{1}{2\mu} \Vert \nabla f(\delta)\Vert_2 \leq \frac{G}{2\mu}
    \]
    From the first-order condition for strong concavity, we have that
    \[
        f(\delta_H^\star) \le f(\delta_F^\star) + \nabla f(\delta_F^\star)^T(\delta_H^\star - \delta_F^\star) - \frac{\mu}{2}\Vert \delta_H^\star - \delta_F^\star\Vert_2^2
    \]
    Using the fact that $\nabla f(\delta_F^\star) = 0$, we get the second bound:
    \[
        \| \delta_F^\star - \delta_H^\star\|_2^2 
        \leq \frac{2}{\mu}\left(f(\delta_F^\star) - f(\delta_H^\star)\right) 
        \le \frac{G}{\mu^2}.
    \]
    For any $\mu$-strongly concave function $f$ with optimum $x^{\star}$, note that for any $y \in \Dom f$
    \[
    f(y) \leq f(x^{\star}) + \nabla f(x^*)^T(y-x^{\star}) - \mu \Vert y-x^{\star} \Vert_2^2 = f(x^{\star}) - \mu \Vert y-x^{\star} \Vert_2^2 
    \]
    which implies that $\mu \Vert y-x^{\star} \Vert_2^2 \leq f(x^{\star}) - f(y)$.
    Applying this to $f$ gives the bound on $\Vert\delta_F^{\star} - \delta_H^{\star}\Vert_2$.
\end{proof}
This bound implies that we can approximate the optimum of $\delta_F^{\star}$
using $\delta_H^{\star}$, which is easily computable. Next, define $S(\delta) =
f(\delta)(p^T \delta)$, where we assume that $p$ is fixed. Let $\delta_S^{\star} \in
\argmax_{\delta} S(\delta)$. We show that the approximate optimum of $f$, $\delta^{\star}_H$ and
$\delta^{\star}_S$ are close:
\begin{claim}
    Suppose the hypotheses of Claim 3.1 hold. Then we have that
    \[
        \|\delta_S^\star - \delta_H^\star\| \le 4G/\mu.
    \]
\end{claim}
\begin{proof}
    Since Claim 3.1 bounds the distance between $\delta_F^\star$ and any $\delta$, apply this claim to both $\delta_H^\star$ and $\delta_S^\star$, then use the triangle inequality.
\end{proof}

Finally, we can bound the objective difference between the approximate and true optimal values, which bounds the profit lost by only approximately solving~\eqref{eq:creation-arb}.
\begin{claim}
    Suppose the hypotheses of Claim 3.1 hold and that $\min_{i} p \geq C$.
    Then we have 
    \begin{align*}
    S(\delta_S^{\star}) - S(\delta_F^{\star}) \leq \left(40 + \frac{8}{C}\right)\Vert p \Vert_2 \frac{G^3}{\mu^2}
    \end{align*}
\end{claim}
\begin{proof}
    Firstly note that the gradient condition implies $f$ is Lipschitz.
    Since $f(0) = 0$, this implies $|f(\delta)| \leq G \Vert \delta \Vert_2$ for all $\delta$.
    Using Cauchy-Schwarz, we have $p^T \delta \leq \Vert p\Vert_2 \Vert \delta \Vert_2$.
    From strong concavity, we have $f(\delta) \leq f(\delta_F^{\star}) - \frac{\mu}{2} \Vert \delta_F^{\star} - \delta\Vert_2^2$.
    This implies that $\Vert \delta_F^{\star} - \delta \Vert \leq \sqrt{\frac{2}{\mu} (f(\delta_F^{\star}) - f(\delta))} \leq \sqrt{\frac{2}{\mu} f(\delta_F^{\star})}$.
    If $\delta = 0$, this is equivalent to $\frac{\mu}{2}\Vert \delta_F^{\star}\Vert^2 \leq f(\delta_F^{\star}) \leq G \Vert \delta_F^{\star}\Vert$, which implies $\Vert \delta_F^{\star}\Vert \leq \frac{2G}{\mu}$, $\Vert \delta_F^{\star}-\delta\Vert_2 \leq \frac{4G}{\mu}$, and $f(\delta^{\star}_F) \leq \frac{2G^2}{\mu}$.
    Therefore we have
    \begin{align*}
    S(\delta^{\star}_S) - S(\delta^{\star}_F) &\leq |S(\delta^{\star}_S) - f(\delta^{\star}_S)| + |S(\delta^{\star}_F) - f(\delta^{\star}_F)| \\
    &\leq f(\delta^{\star}_S) |p^T\delta_S^{\star} + 1| + f(\delta^{\star}_F) |p^T\delta_F^{\star} + 1| \\
    &\leq G \Vert \delta^{\star}_S\Vert_2\left(1 + \Vert p\Vert_2 \Vert \delta_S^{\star}\Vert_2 \right)
    + G \Vert \delta^{\star}_F\Vert_2\left(1 + \Vert p\Vert_2 \Vert \delta_F^{\star}\Vert_2 \right) \\ 
    & \leq G(\Vert \delta^{\star}_S - \delta^{\star}_F\Vert_2 + \Vert \delta^{\star}_F\Vert_2)\left(1+ \Vert p\Vert_2 (\Vert \delta_S^{\star}-\delta^{\star}_F\Vert_2 + \Vert \delta^{\star}_F\Vert_2)\right) + G \Vert \delta^{\star}_F\Vert_2\left(1 + \Vert p\Vert_2 \Vert \delta_F^{\star}\Vert_2 \right) \\
    &= G\left(\frac{6G}{\mu}\right)\left(1 + \Vert p \Vert_2 \frac{6G}{\mu}\right) + G\left(\frac{2G}{\mu}\right)\left(1 + \Vert p \Vert_2 \frac{2G}{\mu}\right) \\
    &\leq \left(40 + \frac{8}{\Vert p \Vert_2}\right) \Vert p \Vert_2 \frac{G^3}{\mu^2}
    \end{align*}
\end{proof}
These claim show that if we find a trade that only approximately optimizes $F$, (possibly via the quadratic proxy,
$H$), we are close to the optimal trade (the optimum of~\eqref{eq:creation-arb}) and the difference in profit between these two trades is bounded. From the approximate optimum, one can utilize a local
method such as (sub)gradient descent to do further refinement.

\subsection{Share redemption}\label{subsec:redemption} An arbitrageur may
alternatively buy PDLP shares on the market and redeem them for a pro-rata share
of the portfolio value. This direction of arbitrage can be profitable even when
the discount is negative. If the arbitrageur buys $\sigma$ PDLP shares, which
has $S$ total shares outstanding, they may redeem these shares for a portfolio
$\lambda \in \reals_+^n$ satisfying
\begin{equation}\label{eq:lambda-constraint}
    p^T\lambda = (1 + f(\lambda)) \cdot \frac{\sigma}{S}V(p, R) 
    = (1 + f(\lambda)) \cdot \frac{\sigma}{S} \cdot p^TR
\end{equation}
where we use $f(\lambda)$ as shorthand to denote $F(p, w^\star, R, -\lambda)$.
In other words, the PDLP requires the value of the shares redeemed, measured by
the price $p$, to equal the LP's pro-rata share of the PDLP reserves, modified 
by the discount rate. 

\paragraph{Arbitrage problem.}
Assume that the external market price equals the implied price (sometimes called
the virtual price) of the PDLP shares, \ie,
\[
    p^\mathrm{mkt} = \frac{p^TR}{S}.
\]
This means that an arbitrage only exists when the discount $f(\lambda)$ is positive.
The arbitrageur aims to redeem PDLP shares
for the most valuable basket of assets $\lambda$. They may not withdraw more
than the available assets (those not used for loans) $R^A$ from the PDLP. This
arbitrage problem can be written as
\[
\begin{aligned}
    &\text{maximize}_{\sigma, \lambda} && p^T\lambda - p^\mathrm{mkt} \sigma\; \\
    &\text{subject to} && p^T\lambda = (1+f(\lambda)) \cdot \sigma \cdot \frac{p^TR}{S}\\
    &&& 0\leq \lambda \leq R^A \\
    &&&  0\leq \sigma \leq S.
\end{aligned}
\]
Substituting in equation~\eqref{eq:lambda-constraint} for $\sigma$, and using our market price assumption, this problem becomes
\begin{equation}\label{eq:redemption-arb}
\begin{aligned}
    &\text{maximize}_{\lambda} && p^T\lambda \cdot \frac{f(\lambda)}{1 + f(\lambda)} \\
    &\text{subject to} && 0\leq \lambda \leq R^A.
\end{aligned}
\end{equation}

\paragraph{Approximate optimization.}
Similar to the creation problem, we can find a surrogate function that approximates the objective.
Using strong concavity, we have that
\[
p^T\lambda \cdot \frac{f(\lambda)}{1 + f(\lambda)} 
\le p^T\lambda \cdot f(\lambda) 
\le p^T\lambda \left(g^T\lambda - \mu\|\lambda\|^2\right),
\]
where $g = \nabla f(\lambda)$. We can now approximately optimize the redemption
problem~\eqref{eq:redemption-arb} similarly to the creation
problem~\eqref{eq:creation-arb}.

\paragraph{Price impact.}
Some PDLPs enforce their discounting mechanism via a protocol-controlled
secondary market where users can trade PDLP shares for the num\'eraire. These
exchanges, such as Jupiter's JupiterSwap~\cite{Jupiter-JLP}, give users
discounts on the fees that they pay to purchase a PDLP share from the secondary
market. This mechanism is similar to having a discount rate and a price 
impact for the share purchase.

\subsection{TWMs reduce volatility for LPs}
In addition to ensuring asset diversity, TWMs reduce the
volatility of LP positions. This benefit permits the associated tokens to be used as
higher-quality collateral elsewhere in the DeFi ecosystem.

\paragraph{Portfolio dilution.}
The portfolio value of the PDLP share in one time
period is the sum of several components: the assets in the PDLP, the fees earned
from loans, and the loss from subsidizing the discount rate. Given a starting
portfolio of assets $R$ with lent assets $\ell$, the portfolio value after a
single TWM update of size $\delta$ is
\[
V^\mathrm{new} = 
\frac{1}{1 + F(p, w^\star, R, \delta)}\cdot p^TR + f p^T\ell.
\]
The first term captures the diluted value of the initial portfolio for existing
LPs, and the second term captures the fees earned from lending.

\paragraph{PDLPs improve portfolio delta.}
The \emph{delta} of a portfolio measures the sensitivity of the portfolio value
to price changes in each of the underlying assets. Here, we measure this via the
gradient of the portfolio value with respect to the asset prices. Passive LPs
typically aim to have low delta portfolios. We show that the difference between
the delta of the original portfolio, $\nabla_p V^\mathrm{old} = R$, and that of
the new portfolio, $\nabla_p V^\mathrm{new}$, is bounded in the following claim:
\begin{claim}\label{claim:delta-bounded} Suppose that $\nabla_p F \geq \frac{8 f
    R}{p^T R}$ and $F \leq 1$. Then we have
    \[
    \nabla_p V^\mathrm{new} \le \left(\frac{1}{2} - f\right) \cdot \nabla_p V^\mathrm{old}.
    \]
\end{claim}

\paragraph{Proof.}
We have that
\[
\begin{aligned}
    \nabla_p V^\mathrm{new}
    &= \frac{1}{1 + F(p,w^\star, R, \delta)} \cdot R  -\frac{p^TR}{(1 + F(p,w^\star, R, \delta))^2} \cdot \nabla_p F  + f\ell \\
    &\le \left(\tfrac{1}{2} + f\right)\cdot R -\frac{p^TR}{(1 + F(p,w^\star, R, \delta))^2} \cdot \nabla_p F \\
    &\le \left(\tfrac{1}{2} + f\right)\cdot R - \frac{1}{4} \cdot 8fR \\
    &= \left(\tfrac{1}{2} - f\right) R
\end{aligned}
\]
In the first inequality, we used that $F \le 1$ and $\ell \le R$, and in the
second we used the assumption on the gradient of $F$.

\paragraph{Discussion.}
In words, this claim shows that as long as the discount decays fast enough at
higher prices and there is a minimal amount of lending demand, the excess delta
added from participating in a PDLP is at most $(\frac{1}{2}-f)$ times greater
than that of holding the portfolio $R$. Note the upper bound on PDLP delta
decreases as fees increase, which suggests that the worst-case risky exposure
the protocol takes decreases at higher fees. This is often argued to be a reason
the JLP style assets are good collateral for lending protocols, where there
exist nearly \$500 million dollars of outstanding JLP loans on
Kamino.~\cite{marius1, kamino-risk}.\footnote{As a concrete example, suppose we
have a lending protocol where where are two assets $A_1, A_2$ that the protocol
offers loan-to-values (LTVs) of $L_1, L_2 \in (0, 1)$. This means that a user
can borrow up to $L_i$\% percent of the value of their $A_i$ collateral in
num\'eraire. A portfolio with weights $(w, 1-w)$ for $A_1, A_2$ can borrow up to
$L(w) = wL_1 + (1-w)L_2 = w(L_1 - L_2) + L_2 \in (0, 1)$ percent of the
portfolio value. Claim~\ref{claim:delta-bounded} implies that a lending protocol
(under no arbitrage) can offer an LTV of $\frac{1}{\frac{3}{2}-f} L(w)$ and have
at most the same price risk as the portfolio $(w, 1-w)$. }

\subsection{Example}
Recall that the weight of a PDLP with portfolio $R$ and asset prices $p$ is
$w(p, R) = ({p \odot R})/{p^T R}.$ GMX's GLP pool~\cite{gmx-repo} has the
discount function
\begin{equation}\label{eq:gmx-discounting}
    F(p, w^{\star}, R, \delta) = 
    \max(0, \gamma_b + \max_{i \in [n]} G_i(w(p, R),\, w(p, R+\delta),\, w^{\star}))
\end{equation}
where $G(w^b, w^a, w^{\star})$ is the piecewise function
\[
G_i(w^b, w^a, w^\star) = \begin{cases}0 & \mathrm{if}\ w^a_i = w^b_i\\ \gamma_t\left\vert\frac{w^b_i - w^*_i}{w^*_i}\right\vert & \mathrm{if}\ \vert w^a_i - w^*_i\vert < \vert w^b_i - w^*_i\vert \\ -\frac{\gamma_t}{2}\left(\left\vert\frac{w^b_i - w^*_i}{w^*_i}\right\vert + \left\vert\frac{w^a_i - w^*_i}{w^*_i}\right\vert\right) & \mathrm{else}\\ \end{cases}
\]
and $\gamma_b, \gamma_t \in (0, 1)$ are called the base and tax scalars.
We plot this function in Figure~\ref{fig:gmx-discount}.
We note that while this function is not $\mu$-strongly concave, one can apply the proximal gradient method~\cite{parikh2014proximal} to construct a $\mu$-strongly concave approximation of this function for which our results apply.

\begin{figure}
    \centering
    \includegraphics[width=0.5\linewidth]{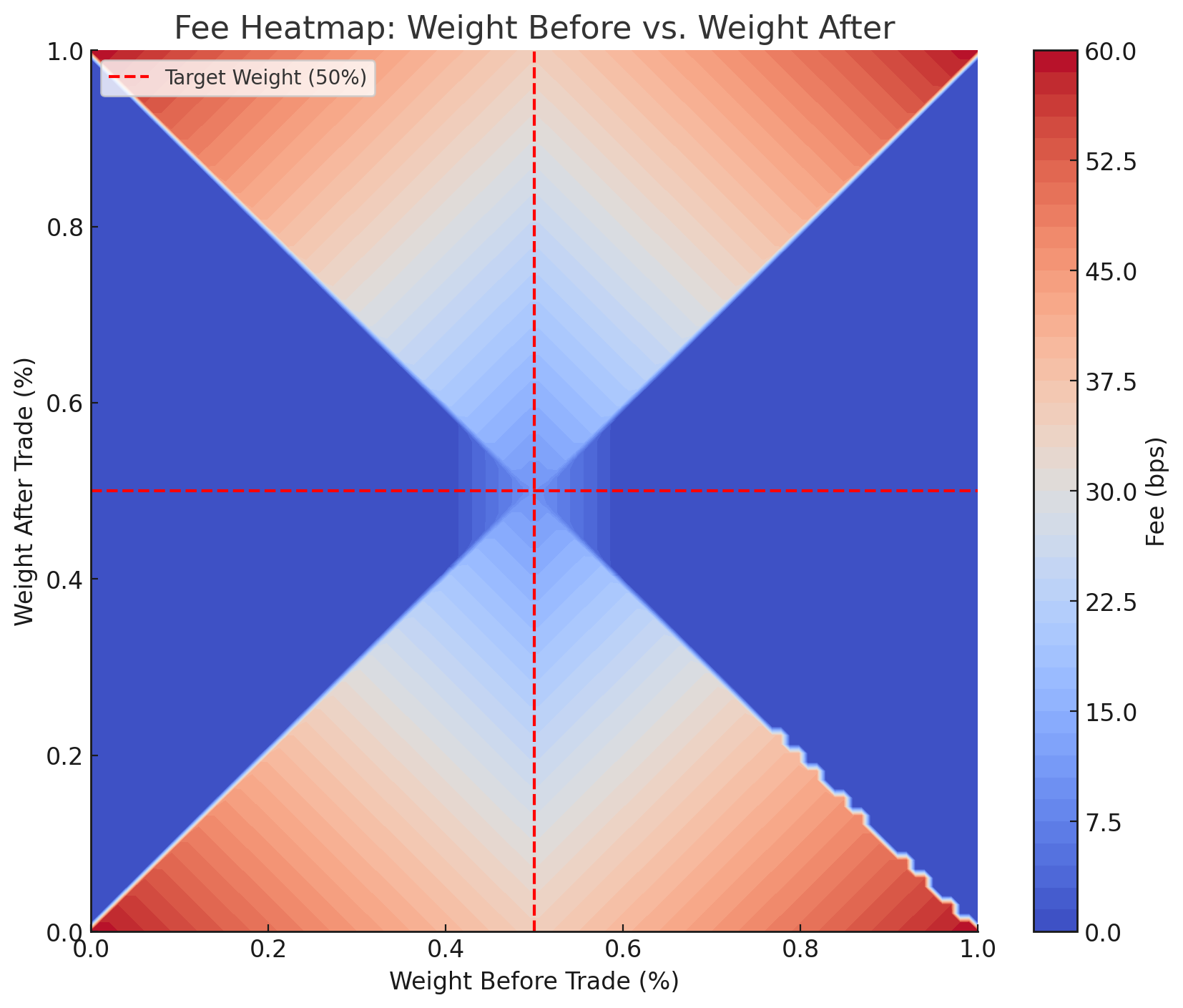}
    \caption{Heatmap of the GMX GLP discount function with two assets for a target weight $w^{\star} = 0.5$}
    \label{fig:gmx-discount}
\end{figure}


\section{Hedged PDLPs}\label{sec:hedged-pdlp}
Empirical observations suggest that liquidity providers can more easily hedge
their positions on PDLPs than on CFMMs. In this section, we study delta hedging
of PDLP LP positions: removing volatility in an LP position by shorting the 
volatile assets in the PDLP portfolio. A delta-hedged LP position aims to only
earn fees in the num\'eraire. In contrast to CFMM positions, LPs can cheaply 
approximately delta hedge PDLP positions, which likely contributed to the rapid 
rise of PDLP LP vaults (\eg,~\cite{Gauntlet-hJLP}).

\paragraph{Delta hedging.}
The delta of an $n$-asset portfolio, denoted by $\Delta \in \reals^n$, measures
the sensitivity of the portfolio value to price changes in each of the
underlying assets. (For a practical example, see~\cite[\S4.4]{Gauntlet-hJLP}.)
Given a portfolio of $n$ risky assets $R \in \reals^{n}_+$, a delta hedge is an
offsetting portfolio $\pi \in \reals^n$ so that the combined portfolio is less
sensitive to price fluctuations (the delta) than the original portfolio. If the
PDLP has $\ell \in \reals^{n}_+$ lent out of the $n$ assets, the the delta hedge
portfolio under the mean-variance
framework~\cite{madan2016adapted,hull2017optimal} is updated by solving the
following optimization problem:
\begin{equation}\label{eq:delta-hedge}
\pi^\mathrm{new} = \argmax_{x} \;
\underbrace{f \cdot {\ell}^T (x + R)}_{\text{fees}} - 
\underbrace{\frac{1}{2}(x - \pi)^T\mathsf{Diag}(c)(x - \pi)}_{\text{rebalancing cost}} 
- \underbrace{\frac{\gamma}{2}(x + \Delta)^T\Sigma(x + \Delta)}_{\text{risk}},
\end{equation}
where $\Sigma \succ 0$ is the covariance matrix for the risky assets, $c \in
\reals^{n}_+$ is the vector of rebalancing costs, and $\gamma \in \reals_+$ is a
risk aversion parameter. This strategy aims to maximize the fees earned from
loans while minimizing the risk of the delta-hedged portfolio, measured by the
variance. The parameter $\gamma$ controls the tradeoff between these two
objectives. The exact solution to this problem is
\[
    \pi^\mathrm{new} =
    \left(\gamma \Sigma + \mathsf{Diag}(c)\right)^{-1}
    \left[f \cdot {\ell} + \mathsf{Diag}(c)\pi - \gamma \Sigma\Delta \right].
\]
Without transaction costs, this expression simplifies to
\begin{equation}\label{eq:delta-hedge-portfolio}
\pi^\mathrm{new} = (f / \gamma) \Sigma^{-1} {\ell} - \Delta.
\end{equation}
For simplicity, we will prove our main result without transaction costs.

\subsection{When is a hedged PDLP's Sharpe ratio improved?}\label{subsec:sharpe}
We assess the performance of this portfolio using it's Sharpe ratio: the
expected return per unit of standard deviation (\ie, risk). In this subsection,
we give sufficient conditions such that the delta-hedged portfolio's Sharpe
ratio is improved over the unhedged portfolio. We prove this result in two
parts. First we show that the expectation of the delta-hedged portfolio is
non-decreasing. Second, we show that the variance of the delta-hedged portfolio
is non-increasing. These two facts together imply that the Sharpe ratio of the
delta-hedged portfolio is no worse than the unhedged portfolio.

\paragraph{Expectation is non-decreasing.}
Here, we derive conditions under which the expectation of the delta-hedged
portfolio is at least as large as the unhedged portfolio. Since the
delta-hedged portfolio is an additive offsetting portfolio, we aim to see when
\[
    \Expect\left[{p^T\pi}\right] 
    = \Expect\left[{(f/\gamma)p^T\Sigma^{-1}\ell - p^T\Delta}\right] 
    \ge 0,
\]
where $\pi$ is the solution to~\eqref{eq:delta-hedge-portfolio}. If 
$\lambda_\mathrm{min}$ and $\lambda_\mathrm{max}$ are the minimum and maximum
eigenvalues of the covariance matrix $\Sigma$, then we have the elementary
inequality
\[
    \lambda_\mathrm{max}^{-1} \cdot p^T\ell 
    \le p^T\Sigma^{-1}\ell.
\]
This implies that the sufficient condition for non-decreasing expectation is
\begin{equation}\label{eq:lend-delta-ineq}
   f \cdot \Expect\left[{p}^T {\ell}\right] \geq \gamma \lambda_{\max} \cdot \Expect\left[{p}^T{\Delta}\right]
\end{equation}
This condition states that as long as the fee revenue, driven by loan demand, is $\gamma \lambda_{\max}$
times the portfolio delta, the expectation is non-decreasing. This implies that
the correlation between the assets in the pool drives how profitable one needs
to be for delta hedging to work. We discuss when a PDLP should be split into
multiple pools in Appendix~\ref{app:multipool}.

\paragraph{Variance is non-increasing.}
Next, we derive conditions under which the variance of the delta-hedged
portfolio is at most as large as that of the unhedged portfolio. The variance
of the delta-hedged portfolio can be written as
\[
    \Var\left[p^T\pi + p^TR\right]
    =  \Var\left[p^T\pi\right] + 2\mathbf{Cov}\left(p^T\pi, p^TR\right) + \Var\left[p^TR\right].
\]
Thus, we must have that
\[
    \Var\left[p^T\pi\right] \le - 2\mathbf{Cov}\left(p^T\pi, p^TR\right).
\]
Using the fact that the delta hedge in negatively correlated to the unhedged
portfolio value and that covariance is bounded by the product of the standard
deviations, we have the sufficient condition
\[
    \Var\left[p^T\pi\right] \le 4 \cdot \Var\left[p^TR\right].
\]
In other words, as long as the hedge is anticorrelated to the portfolio and that
its variance is less than 4 times the variance of the portfolio, variance is
non-increasing under the delta hedge.

\paragraph{Main result.}
Putting the two conditions together, we have conditions under which the Sharpe
ratio of the delta-hedged portfolio is at least as good as the unhedged
portfolio. We summarize this result in the following claim:
\begin{claim}
Suppose we have a PDLP with assets $R \in \reals_+^n$ which have covariance
$\Sigma \prec \lambda_\mathrm{max} I$, loan fees $f$, and loan demand $\ell$.
Denote the asset prices by $p \in \reals_{++}^n$ and set the risk parameter
$\gamma \in \reals_+$. Then the Sharpe ratio of the delta-hedged portfolio,
which includes the original portfolio plus an offseting portfolio $\pi \in
\reals^n$, is at least as high as the unhedged portfolio if the following
conditions hold:
\begin{enumerate}
    \item The loan fees are at least $\gamma \lambda_{\max}$ times the portfolio delta:
    \[
        f \cdot \Expect[p^T\ell] \geq \gamma \lambda_{\max} \cdot \Expect[p^T\Delta].
    \]
    \item The variance of the offsetting portfolio is at most 4 times the 
    variance of the unhedged portfolio:
    \[
        \Var[p^T\pi] \leq 4 \cdot \Var[p^TR].
    \]
\end{enumerate}
\end{claim}

\section{Conclusion and future work}
We presented the first formalization of Perpetual Demand Lending Pools (PDLPs),
which hold \$2.5 billion in assets as of December 2024 and have generated over
\$700 million in fees by lending to traders on decentralized perpetuals
exchanges. Our formulation of the arbitrage problems for various counterparties
suggests a wide, largely-unexplored design space for such protocols. The scale
and rapid growth of PDLPs begs a question: Can novel PDLPs service the same
notional volume with less capital?

Our framework provides a number of directions to explore. First, formalizing
multi-period PDLP arbitrage could inform the design of dynamic fees and
discounting mechanisms. Such multi-period models exist for CFMMs (\eg~via the
quantification of loss-versus-rebalancing). Second, dynamic target portfolios
may improve target weight mechanisms. Such adjustments have been proposed for
CFMMs and lending protocols, but these mechanisms will need different properties
for perpetuals protocols. Finally, the plethora of hedged and levered PDLP
strategies (especially for Jupiter~\cite{drift-vaults}) suggest that levered
PDLPs present an interesting avenue of future exploration. We hope that this
work helps researchers improve these mechanisms and make decentralized exchanges
competitive with their centralized counterparts.

\section{Acknowledgments}
We want to thank Tim Copeland, Alex Evans, Kshitij Kulkarni, Horace Pan, Krane, Victor Xu, JD Maturen, Diogenes Casares, Marius Ciubotariu, and Madison Piercy for helpful comments and suggestions.

\bibliographystyle{alpha}
\bibliography{bib.bib} 

\appendix
\input{dynamics}
\input{multiple-pools}

\end{document}

%% file: defs.tex
\newcommand{\ones}{\mathbf 1}
\newcommand{\reals}{{\mbox{\bf R}}}

\newcommand{\naturals}{{\mbox{\bf N}}}

\newcommand{\Expect}{\mathop{\bf E{}}}
\newcommand{\Var}{\mathop{\bf Var{}}}

\newcommand{\Dom}{\mathop{\bf Dom}}

\newcommand{\Prob}{\mathop{\bf Prob}}

\newcommand{\argmax}{\mathop{\rm argmax}}

\newcommand{\eg}{{\it e.g.}}
\newcommand{\ie}{{\it i.e.}}

\newcommand{\BEAS}{\begin{eqnarray*}}
\newcommand{\EEAS}{\end{eqnarray*}}
\newcommand{\BEA}{\begin{eqnarray}}
\newcommand{\EEA}{\end{eqnarray}}
\newcommand{\BEQ}{\begin{equation}}
\newcommand{\EEQ}{\end{equation}}
\newcommand{\BIT}{\begin{itemize}}
\newcommand{\EIT}{\end{itemize}}

%% file: dynamics.tex
\section{Trading Dynamics}\label{app:dynamics}
We more formally describe the update dynamics for a perpetuals exchange using a PDLP in this section.
For each time $t \in \naturals$, let $\tau_t = \{(c^i_t, \delta^i_t, \eta^i_t, p_t)\}$ be the set of trades created by users at time $t$ and let $\mathcal{A}_t = \cup_{s \leq t} \tau_s$ be the set of all trades created at or before time $t$.
We define the set of liquidatable long and short trades at time $t$ as
\begin{align*}
\ell_L(t, p_1, \ldots, p_t) &= \left\{ (c^i_s, \delta^i_s, \eta^i_s, p_s) \in \mathcal{A}_t : \eta^i_s > 0, \wedge \forall s' \in [s, t) \; p_{s'} > L(\eta, p_s) \wedge p_t \leq L(\eta, p_s)  \right\} \\
\ell_S(t, p_1, \ldots, p_t) &= \left\{ (c^i_s, \delta^i_s, \eta^i_s, p_s) \in \mathcal{A}_t : \eta^i_s < 0, \wedge \forall s' \in [s, t) \; p_{s'} < L(\eta, p_s) \wedge p_t \geq L(\eta, p_s)  \right\}
\end{align*}
In words, this is the set of trades at times $s < t$ such that the trade was not liquidatable at any time $s'$ with $s \leq s' < t$ but was liquidatable at time $t$.
That is, the trades in these sets are first liquidatable at time $t$.
We define the set $\mathcal{T}_t$ of valid trades at time $t$ to be
\[
\mathcal{T}_t = \mathcal{A}_t - \bigcup_{s \leq t} (\ell_L(t, p_1, \ldots, p_t) \cup \ell_S(t, p_1, \ldots, p_t) )
\]
Finally, we let $\mathcal{L}_t, \mathcal{S}_t \subset \mathcal{T}$ be the set of long and/or short positions alive at time $t$.
Given these definitions, we can describe the dynamics formally as follows:
\begin{enumerate}
    \item The price is updated by the price oracle to $p_{t+1}$
    \item The liquidatable positions are removed:
    \begin{align*}
    \mathcal{L}_{t} &\leftarrow \mathcal{L}_t - \ell_L(t, p_1, \ldots, p_{t+1}) &&&
    \mathcal{S}_t &\leftarrow \mathcal{S}_t - \ell_S(p_1, \ldots, p_{t+1})
    \end{align*}
    \item For each market, the funding rate $\gamma^i_{t}$ for the previous period is paid out, where
    \[
    \gamma^i_{t} = \gamma_L(L^i_t, S^i_t, p_{t+1}, p_t) = \kappa\left(\frac{L^i_t}{S^i_t} - \frac{p_{t+1}}{p_t}\right)
    \]
    and the long and short positions are updated as
    \begin{align*}
    L^i_t &\leftarrow L^i_t + \gamma^i_t S^i_t \ones_{\gamma^i_t \geq 0} &&&
    S^i_t &\leftarrow S^i_t + \gamma^i_t L^i_t \ones_{\gamma^i_t \leq 0}
    \end{align*}
    Each trader $j \in [M]$ receives a pro-rata percentage of the funding rate based on the relative size of their position.
    Note that the payout for the previous period occurs before the new trades for the following period are added.
    \item LPs submit updates to their portfolio, $q^i_t \in \reals^n$, which modifies the current portfolio $R_t$ to $R_{t+1} = R_t + \sum_{i=1}^K q^i_t$ subject to the constraint that $R_t + \sum_{i=1}^{\ell} q^i_t \geq 0$ for all $\ell \in [k]$.
    This constraint ensures that the LPs cannot drain the pool.
    \item For each trader $j \in [M]$ that submits a trade $(c^i_t, \delta^i_t, \eta^i_t, p_t) \in \tau_t$:
    \begin{itemize}
        \item Collateral $c_t^i$ is chosen such that $(1+\gamma)|\eta_t^i| p_{t+1}^T c^t_i = p_{t+1}^T \delta_t^i$.
        If this is not possible due to the user's previous positions, the trade is rejected
        \item If $c_{L,t} + c_{S, t} + \sum_{i \leq j} c_t^i > R_{t+1}$, the trade from trader $j$ is rejected as it violates the solvency requirement
        \item Otherwise update $\mathcal{L}_{j, t}, \mathcal{S}_{j, t}$ via
        \begin{align*}
            \mathcal{L}_{j, t+1} &= \mathcal{L}_{j, t} + \ones_{\eta^j_t >0 } (c^j_t, \delta^j_t, \eta^j_t) &&&
            \mathcal{S}_{j, t+1} &= \mathcal{S}_{j, t} + \ones_{\eta^j_t < 0} (c^j_t, \delta^j_t, \eta^j_t)
        \end{align*}
    \end{itemize}
    \item Fees from the last period are paid out to liquidity providers from the traders
    \[
    R_{t+1} \leftarrow R_{t+1} + f\cdot(c_{L, t+1} + c_{S, t+1})
    \]
\end{enumerate}

%% file: multiple-pools.tex
\section{When is it better to have multiple pools?}\label{app:multipool}
Recall that GMX moved from a single-pool PDLP model to a multiple pool PDLP system.
In this system, a PDLP portfolio $R$ is partitioned into sub-portfolios $R = \sum_{i=1}^k R_i$ where $R_i$ represents the assets in pool $i \in [k]$.
Each pool services a different set of perpetuals with a unique target weight, among other parameters.
A natural question to ask is when is it more efficient to aggregate the $k$ pools into a single pool.

We will consider the simplest model to analyze this question with $k = 2$ and $\mathsf{supp}(R_1) \cap \mathsf{supp}(R_2) = \emptyset$, \ie,~the two PDLP pools have no assets in common.
We will assume that $|\mathsf{supp}(R_1)| = m_1, |\mathsf{supp}(R_2)| = m_2$ and that $m_1+m_2 = n$.
We will abuse notation slightly as also refer to the set of assets in $R_i$ as $R_i$.
We consider the full symmetric stochastic covariance matrix $\Sigma \in \reals^{n \times n}$ and write
\[
\Sigma = \left[
\begin{matrix}
    A & B \\ 
    B^T & C
\end{matrix}
\right]
\]
where $A \in \reals^{m_1 \times m_1}, B \in \reals^{m_1 \times m_2}, C \in \reals^{m_2 \times m_2}$ are matrices with $A, C$, non-singular and positive definite.

For mean-variance optimization, one computes the conditional covariance\footnote{Note that this is technically only the conditional covariance for multivariate normal distributions; for generic distributions it corresponds to the partial correlation~\cite[\S6.2.3]{zhang2006schur}.} using Schur Complements~\cite{ouellette1981schur, zhang2006schur}.
We denote the Schur complements as $\Sigma / A = A - B C^{-1} B^{T}$ and $\Sigma / C = C - B^T A^{-1} B$.
$\Sigma / A$ represents the covariance matrix after marginalizing the variables in $R_2$, and similarly for $\Sigma / C$~\cite[\S6.2.3]{zhang2006schur}.
For a covariance matrix $S$, we define the portfolio values $V_t^S$ as the delta hedged portfolio values when the covariance matrix is $S$.

Our goal is to find conditions under which the delta hedged portfolio for a single pool is better than an isolated pool.
We can view the isolated pool as having either the covariance $A$ or $C$ without the impact of the other, whereas the single pool needs to incorporate information from both pools.
We can view the conditional covariance matrices $\Sigma / A$ and $\Sigma / C$ as incorportating information from both pools and hence solving the mean-variance problem for these matrices represents the single pool solution.

If the single pool provides better returns than multiple pools, this corresponds to the following conditions: $\Expect[V_t^{\Sigma / A} - V(\mathcal{P}, p)] \geq \Expect[V_t^{A}-V(\mathcal{P}, p)]$ and $\Expect[V_t^{\Sigma / C} - V(\mathcal{P}, p)] \geq \Expect[V_t^{C}-V(\mathcal{P}, p)]$
which reduces to
\begin{align}\label{eq:single-pool-win-condition}
     \Expect[V_t^{\Sigma / A}] \geq \Expect[V_t^{A}] &&    \Expect[V_t^{\Sigma / C}] \geq \Expect[V_t^{C}]
\end{align}
We prove the following sufficent condition for when this occurs:
\begin{claim}\label{claim:schur}
    Suppose that $\sigma_{\min}(\Sigma_X) > \sigma_{\max}(A)$ and $\sigma_{\min}(\Sigma_Y) > \sigma_{\max}(B)$.
    Then~\eqref{eq:single-pool-win-condition} holds and a single pool provides better delta hedged returns than multiple pools.
\end{claim}
\begin{proof}
Conditions~\eqref{eq:single-pool-win-condition} are equivalent to
\begin{equation}\label{eq:delta-portfolio}
\Expect[p^T(1, \pi_t^{\Sigma / A})] \geq \Expect[p^T(1, \pi_t^{A})]
\end{equation}
where $\pi^{\Sigma}_t$ is the optimal portfolio with covariance $\Sigma$.
We can proceed analogously to~\ref{subsec:sharpe}, which makes~\eqref{eq:delta-portfolio} equivalent to the condition
\[
\frac{f}{\gamma} \tilde{p}^T (\Sigma/A) \tilde{\ell} - \tilde{p}^T \tilde{\Delta}(\mathcal{P}) \geq \frac{f}{\gamma}\tilde{p}^T A \tilde{\ell} - \tilde{p}^T \tilde{\Delta}(\mathcal{P})
\]
which is equivalent to
\[
\frac{f}{\gamma}\tilde{p}^T (\Sigma / A) \tilde{\ell} \geq \frac{f}{\gamma}\tilde{p}^T A \tilde{\ell}
\]
A sufficient condition for this to hold for all $\tilde{p}, \tilde{\ell}$ is for $\sigma_{\max}(A) < \sigma_{\min}(\Sigma / A)$.
We note that this doesn't violate Cauchy's interlacing theorem style results which show that $ \sigma_{n-m_1}(\Sigma)< \sigma_{\max}(\Sigma / A) \leq \sigma_{\max}(\Sigma)$
where $\sigma_{k}(X)$ is the $k$th eigenvalue of a matrix $X$ since those are between the spectral of the full matrix $\Sigma$ and the Schur complement as opposed to the subcomponent~\cite[Ch. 2]{zhang2006schur}. 
The same proof applies for the other condition.
\end{proof}

\noindent We note these conditions are similar to hierarchical risk-parity methods~\cite{cotton2024schur, lopez2016building}.
We also note that these conditions suggest that GMX V2's recent dynamic pricing upgrade~\cite{gmx-dynamic-price-impact} is not sufficient alone for ensuring that V2 vaults are easily delta hedgeable as it doesn't take into account pool asset covariance as is done here.
We leave it for future work to connect the dynamic pricing model with the cost of delta hedging.